\documentclass[runningheads]{llncs}

\usepackage{graphicx}
\usepackage{pifont}
\usepackage{verbatim}
\usepackage{amsmath,amsfonts}
\usepackage{enumerate,url}
\usepackage{paralist}
\usepackage{subfig}
\usepackage{wrapfig}

\usepackage{algorithm}
\usepackage[noend]{algpseudocode}

\begin{document}

\let\doendproof\endproof
\renewcommand\endproof{~\hfill\qed\doendproof}

\newcommand{\leftp}[1]{\texttt{left}(#1)}
\newcommand{\rightp}[1]{\texttt{right}(#1)}

\newcommand{\lefte}[1]{\texttt{leftEdge}(#1)}
\newcommand{\righte}[1]{\texttt{rightEdge}(#1)}

\newcommand{\vis}[1]{\texttt{Vis}(#1)}

\spnewtheorem{myremark}{Remark}{\bfseries}{\itshape}
\newcommand{\down}[1]{\raisebox{-.4ex}{#1}}

\makeatletter
    \renewcommand*{\@fnsymbol}[1]{\ensuremath{\ifcase#1\or *\or \dagger\or \S\or
       \mathsection\or \mathparagraph\or \|\or **\or \dagger\dagger
       \or \ddagger\ddagger \else\@ctrerr\fi}}
\makeatother

\newcommand*\samethanks[1][\value{footnote}]{\footnotemark[#1]}

\title{A Note on Approximating 2-Transmitters}

\titlerunning{A Note on Approximating 2-Transmitters}

\author{
Saeed Mehrabi$^1$,
Abbas Mehrabi$^2$
}

\authorrunning{S. Mehrabi and A. Mehrabi}

\institute{
$^1$~Cheriton School of Computer Science, University of Waterloo, Waterloo, Canada.
\email{smehrabi@uwaterloo.ca}\\
$^2$~School of Information and Communications, Gwangju Institute of Science and Technology, Gwangju, South Korea.
\email{mehrabi@gist.ac.kr}
}


\newcommand{\keywords}[1]{\par\addvspace\baselineskip
\noindent\keywordname\enspace\ignorespaces#1}


\maketitle

\begin{abstract}
A \emph{$k$-transmitter} in a simple orthogonal polygon $P$ is a mobile guard that travels back and forth along an orthogonal line segment $s$ inside $P$. The $k$-transmitter can see a point $p\in P$ if there exists a point $q\in s$ such that the line segment $pq$ is normal to $s$ and $pq$ intersects the boundary of $P$ in at most $k$ points. In this paper, we give a 2-approximation algorithm for the problem of guarding a monotone orthogonal polygon with the minimum number of 2-transmitters.
\end{abstract}

\section{Introduction}
\label{sec:introduction}
In the standard version of the art gallery problem, introduced by Klee in 1973~\cite{orourke1987}, we are given a simple polygon $P$ in the plane and the goal is to guard $P$ by a set of point guards. That is, we need to find a set of point guards such that every point in $P$ is seen by at least one of the guards, where a guard $g$ sees a point $p$ if and only if the segment $gp$ is contained in~$P$. Chv\'{a}tal~\cite{chvatal1975} proved that $\lfloor n/3\rfloor$ point guards are always sufficient and sometimes necessary to guard a simple polygon with $n$ vertices. The art gallery problem is known to be \textsc{NP}-hard on arbitrary polygons~\cite{lee1986}, orthogonal polygons~\cite{dietmar1995} and even monotone polygons~\cite{krohn2013}. Eidenbenz~\cite{eidenbenz1998} proved that the art gallery problem is \textsc{APX}-hard on simple polygons and Ghosh~\cite{ghosh2010} gave an $O(\log n)$-approximation algorithm that runs in $O(n^4)$ time on simple polygons. Krohn and Nilsson~\cite{krohn2013} gave a constant-factor approximation algorithm on monotone polygons. They also gave a polynomial-time algorithm for the orthogonal art-gallery problem that computes a solution of size $O(OPT^2)$, where $OPT$ is the cardinality of an optimal solution.

Many variants of the art gallery problem have been studied. Katz and Morgenstern~\cite{katz2011} introduced a variant of this problem in which \emph{$k$-transmitters} are used to guard orthogonal polygons. A $k$-transmitter $T$, where $k\geq 0$, is a point guard that travels back and forth along an orthogonal line segment inside an orthogonal polygon $P$. A point $p$ in $P$ is visible to $T$, if there is a point $q$ on $T$ such that the line segment $pq$ is normal to $T$ and it intersects the boundary of $P$ in at most $k$ points. In the \emph{Minimum $k$-Transmitters (M$k$T)} problem, the objective is to guard $P$ with the minimum number of $k$-transmitters. Katz and Morgenstern introduced the M$k$T problem for only $k=0$ (we remark that 0-transmitters are called \emph{sliding cameras} in~\cite{katz2011}). They first considered a restricted version of the problem, where only vertical 0-transmitters are allowed, and solved this restricted version in polynomial time for simple orthogonal polygons. When both vertical and horizontal 0-transmitters are allowed (i.e., the M0T problem), they gave a 2-approximation algorithm on monotone orthogonal polygons, which was later improved by the $O(n)$-time exact algorithm of de Berg et al.~\cite{deBergDM2014}. Durocher and Mehrabi~\cite{durocherM2013} showed that the M0T problem is \textsc{NP}-hard when $P$ is allowed to have holes. Mahdavi et al.~\cite{mahdaviSG2014} proved that the problem of guarding an orthogonal polygon with $k$-transmitters so as to minimize the total length of line segments along which $k$-transmitters travel is \textsc{NP}-hard for any fixed $k\geq 2$, and gave a 2-approximation algorithm for this problem. To our knowledge, the complexity of the M$k$T problem is open on simple orthogonal polygons for any fixed $k\geq 0$.

We remark that the exact algorithm of de Berg et al.~\cite{deBergDM2014} for the M0T problem on monotone orthogonal polygons does not provide any constant-factor approximation algorithm for the M2T problem. Figure~\ref{fig:deBergAlgorithm} shows a polygon $P$ for which five 0-transmitters are required, but $P$ can be guarded with only one 2-transmitter. Note that the example can be extended to show that an exact solution for the M0T problem does not provide any constant-factor approximation to that of the M2T problem.

\begin{figure}[t]
\centering%
\includegraphics[width=.60\textwidth]{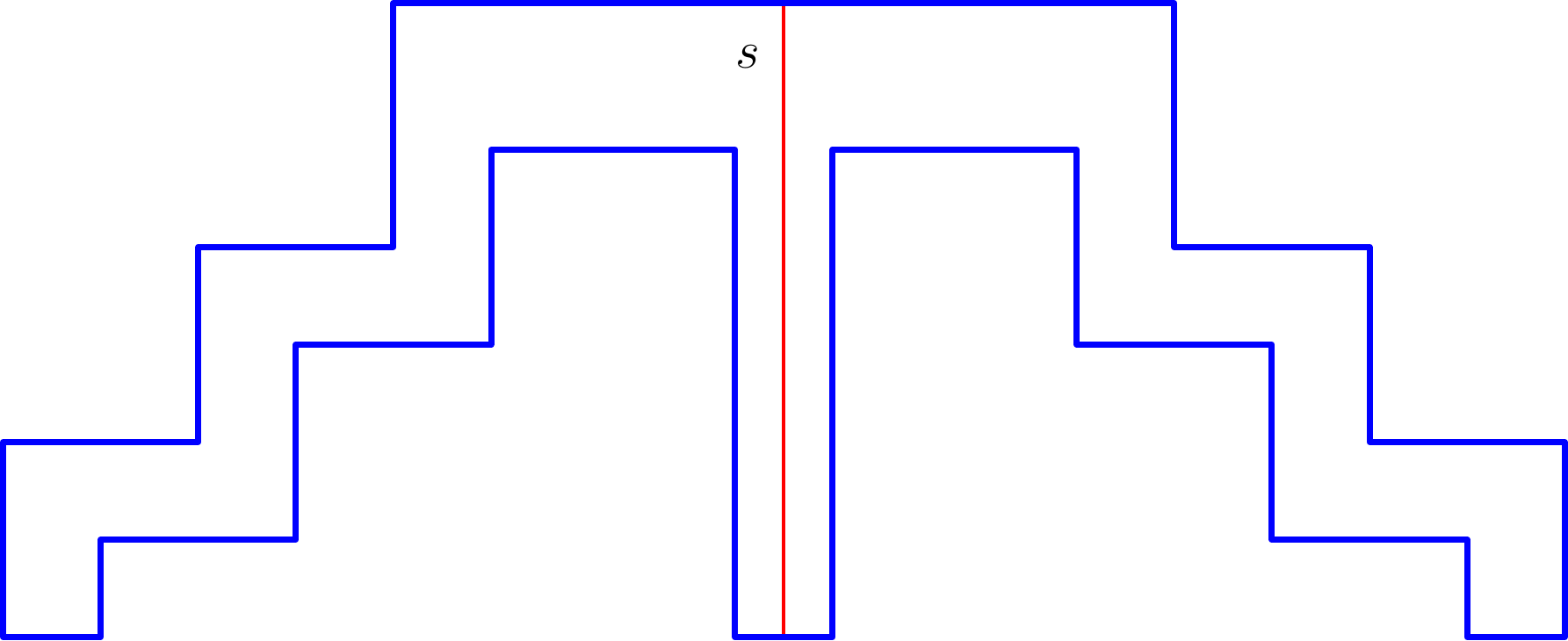}
\caption{A monotone orthogonal polygon $P$ that can be guarded by a single 2-transmitter $s$ while five 0-transmitters are required to guard $P$ entirely. This example can be extended to show that the exact algorithm of de Berg et al.~\cite{deBergDM2014} for the M0T problem does not provide any constant-factor approximation to an exact solution for the M2T problem on $P$.}
\label{fig:deBergAlgorithm}%
\end{figure}

\noindent\paragraph{\bf Our Result.} In this paper, we give a polynomial-time 2-approximation algorithm for the M2T problem on simple and monotone orthogonal polygons. Some preliminaries are given in Section~\ref{sec:prelimin}. We then present our 2-approximation algorithm in Section~\ref{sec:3approximation} and conclude the paper in Section~\ref{sec:conclusion}.


\section{Preliminaries}
\label{sec:prelimin}
Throughout this paper, let $P$ be a simple and $x$-monotone orthogonal polygon with $n$ vertices. A vertex $u$ of $P$ is called \emph{convex} (resp., \emph{reflex}), if the angle at $u$ that is interior to $P$ is 90$^\circ$ (resp., 270$^\circ$). We denote the leftmost and rightmost vertical edges of $P$ (that are unique) by $\lefte{P}$ and $\righte{P}$, respectively. Let $V_P=\{e_1=\lefte{P},e_2,\dots,e_m=\righte{P}\}$, for some $m>0$, be the set of vertical edges of $P$ ordered from left to right. Let $P^+_i$ (resp., $P^-_i$), for some $1\leq i\leq m$, denote the subpolygon of $P$ that lies to the right (resp., to the left) of the vertical line through $e_i$.

Let $s$ be an orthogonal line segment in $P$. We denote the left endpoint and the right endpoint of $s$ by $\leftp{s}$ and $\rightp{s}$, respectively. If $s$ is vertical, we define its left and right endpoints to be its upper and lower endpoints, respectively. Moreover, we denote the $k$-transmitter that travels along $s$ by $s^{(k)}$. For a $k$-transmitter $t$ in $P$, we define the \emph{visibility polygon} of $t$ as the maximal subpolygon $\vis{t}$ of $P$ such that every point in $\vis{t}$ is guarded by $t$.



For each reflex vertex $v$ of $P$, extend the edges incident to $v$ inward until they hit the boundary of $P$. Let $C(P)$ be the set of all maximal line segments in $P$ that are obtained in this way. A \emph{feasible solution} for the M2T problem is a set $M$ of 2-transmitters that guards the entire polygon $P$. A feasible solution $M$ is \emph{optimal} (or, \emph{exact}) if $\lvert M\rvert\leq \lvert S'\rvert$, for all feasible solutions $S'$. We say that a feasible solution $M$ for the M2T problem is in \emph{standard form} if and only if $M\subseteq C(P)$ and every vertical 2-transmitter in $M$ is \emph{vertically maximal}; that is, it extends as far upwards and downwards as possible.

\begin{lemma}
\label{lem:standardOPT}
There exists an optimal solution $OPT^*$ for the M2T problem on $P$ that is in standard form.
\end{lemma}
\begin{proof}
Take any optimal solution $OPT$ for the M2T problem on $P$. First, for each line segment $s\in OPT$ that is not aligned with an edge of $P$, move $s$ vertically up or down, or horizontally to the left or right until it hits an edge of $P$. Next, for every line segment $s'\in OPT$ that is not maximal, replace $s'$ with the maximal line segment in $P$ that is aligned with $s'$. Set $OPT^* :=OPT$. Clearly, $OPT^*$ is a feasible solution for the M2T problem and every line segment in $OPT^*$ is maximal and aligned with an edge of $P$. So, $OPT^*\subseteq C(P)$. Since $\lvert OPT^*\rvert \leq \lvert OPT\rvert$, we conclude that $OPT^*$ is an optimal solution for the M2T problem that is in standard form. This completes the proof of the lemma.
\end{proof}

For a horizontal line segment $t\in P$ and any $k>0$, the visibility polygon of a 0-transmitter that travels along $t$ is the same as that of a $k$-transmitter that travels along $t$. We state and prove this observation more formally.

\begin{lemma}
\label{lem:horizontal0andKTransmitters}
Let $t$ be a horizontal line segment in $P$. Then, $\vis{t^{(0)}}=\vis{t^{(k)}}$ for any $k>0$.
\end{lemma}
\begin{proof}
It is clear that any point in $P$ that is visible to $t^{(0)}$ is also visible to $t^{(k)}$ and so $\vis{t^{(0)}}\subseteq \vis{t^{(k)}}$. Now, let $p$ be a point in $P$ that is visible to $t^{(k)}$. Since $t$ is horizontal and $P$ is an $x$-monotone orthogonal polygon, we conclude that the line segment $pq$ does not intersect the boundary of $P$, where $q$ is the projection of $p$ onto $t$. This means that $p$ is also visible to $t^{(0)}$ and therefore, $\vis{t^{(k)}}\subseteq \vis{t^{(0)}}$. This completes the proof of the lemma.
\end{proof}


\section{A 2-Approximation Algorithm}
\label{sec:3approximation}
In this section, we give our 2-approximation algorithm for the M2T problem on monotone orthogonal polygons. Recall that in the M2T problem, the objective is to guard the polygon $P$ with minimum number of 2-transmitters, where a 2-transmitter can be either horizontal or vertical. For a point $p\in P$, let $L(p)$ denote the vertical line through $p$. We say that a horizontal 2-transmitter in $P$ is \emph{rightward maximal} if it extends as far to the right as possible.

The algorithm initially guards a leftmost portion of the polygon $P$ by two 2-transmitters with different orientations, and then will guard the remaining part of $P$ recursively. The order of the two initial 2-transmitters is determined by whether locating first a vertical 2-transmitter and then a horizontal one would guard a larger portion of $P$ than locating first a horizontal 2-transmitter and then a vertical one. In the following, we describe the algorithm more formally.

\noindent\paragraph{\bf Algorithm.} Let $s_v$ be the rightmost maximal vertical 2-transmitter in $P$ such that every point of $P$ that is to the left of $s_v$ is seen by $s_v$; let $p$ be the leftmost point of $P$ that is not seen by $s_v$. Moreover, let $s_h$ be the rightward maximal horizontal 2-transmitter in $P$ such that $\leftp{s_h}$ lies on $L(p)$. Clearly, $\rightp{s_h}$ lies on a vertical edge $e_i$ of $P$. Observe that $P^-_i$ is entirely guarded by $s_v$ and $s_h$. Given $P$, we define \texttt{vHFinder($P$)} as a method that computes $s_v$ and $s_h$ as described above and returns the triple $(s_v, s_h, e_i)$. Note that \texttt{vHFinder($P$)} guards $P^-_i$ by first locating a vertical 2-transmitter and then a horizontal one from left to right. We next consider the other case.

Let $s'_h$ be the rightward maximal horizontal 2-transmitter in $P$ such that every point of $P$ that is to the left of $L(\rightp{s'_h})$ is seen by $s'_h$. Suppose that $\rightp{s'_h}$ lies on some vertical edge $e_\ell$ ($1\leq \ell\leq m$) of $P$. Let $s'_v$ be the rightmost maximal vertical 2-transmitter in $P$ such that every point of $P$ that lies between $L(\rightp{s'_h})$ and $s'_v$ is guarded by $s'_v$. Moreover, let $p'$ be the leftmost point of $P^+_\ell$ that is not seen by $s'_v$; clearly, $p'$ lies on a vertical edge $e_j$ ($1\leq j\leq m$) of $P$. Observe that $s'_h$ and $s'_v$ guard $P^-_j$ entirely. We now define \texttt{hVFinader($P$)} as a method that computes $s'_h$ and $s'_v$ as described above and returns the triple $(s'_h, s'_v, e_j)$.

\begin{algorithm}[t]
\caption{\textsc{Approximate2Transmitters($P$)}}
\label{alg:approximate2transmitters}
\begin{algorithmic}[1]
\For {each line segment $s\in C(P)$}
	\If {$\vis{s}\subseteq \bigcup_{s'\in C(P)\setminus \{s\}} \vis{s'}$}
		\State $C(P) \gets C(P)\setminus \{s\}$;
	\EndIf
\EndFor
\State $S \gets \emptyset$;
\While {$P \ne \emptyset$}
	\State $(s_v, s_h, e_i)\gets$ \texttt{vHFinder($P$)};\Comment{$\{s_v, s_h\}\subseteq C(P)$}
	\State $(s'_h, s'_v, e_j)\gets$ \texttt{hVFinder($P$)};\Comment{$\{s'_h, s'_v\}\subseteq C(P)$}
	\If {$i>j$}
		\State $S \gets S\cup\{s_v, s_h\}$;
		\State $P\gets P^+_i$;
	\Else
		\State $S \gets S\cup\{s'_h, s'_v\}$;
		\State $P\gets P^+_j$;
	\EndIf
\EndWhile
\State \Return $S$;
\end{algorithmic}
\end{algorithm}

The algorithm is shown in Algorithm~\ref{alg:approximate2transmitters}. In the first step of the algorithm, we remove from $C(P)$ those line segments whose visibility polygon is a subset of the union of the visibility polygons of all other line segments in $C(P)$. Then, in a while-loop, we iteratively \begin{inparaenum}[(i)] \item compute the pairs of 2-transmitters $\{s_v, s_h\}$ and $\{s'_h, s'_v\}$ using the methods \texttt{vHFinder($P$)} and \texttt{hVFinder($P$)}, respectively, and then \item update $P$ depending on whether $i>j$ (i.e., the 2-transmitters $\{s_v, s_h\}$ guard a larger portion of $P$ than $\{s'_h, s'_v\}$) or $j\geq i$ (i.e., the 2-transmitters $\{s'_h, s'_v\}$ guard a larger portion of $P$ than $\{s_v, s_h\}$). \end{inparaenum} We remark here that by Lemma~\ref{lem:standardOPT}, we can assume that both methods \texttt{vHFinder($P$)} and \texttt{hVFinder($P$)} select the 2-transmitters from the set $C(P)$. When $P$ is entirely guarded, we return the set $S$ of 2-transmitters.

\noindent\paragraph{\bf Analysis.} We first note that by Lemma~\ref{lem:standardOPT}, we can assume that the four 2-transmitters computed by \texttt{vHFinder($P$)} and \texttt{hVFinder($P$)} are always in standard form. That is, we restrict our attention to the line segments in $C(P)$ when computing the set $S$. To see the approximation factor of the algorithm, let $P_1, P_2,\dots, P_k$ be the partition of $P$ into $k$ subpolygons ordered from left to right such that the subpolygon $P_i$ is guarded in the $i$th iteration of the while-loop. More precisely, $P_i$ is the subpolygon of $P$ that is cut out from $P$ in the $i$th iteration of the while-loop of the algorithm. It is clear that Algorithm~\ref{alg:approximate2transmitters} locates at most $2k$ 2-transmitters to guard $P$ entirely; that is, $\lvert S \rvert\leq 2k$. In the following, we show that $\lvert OPT\rvert\geq k$ for any optimal solution $OPT$ for the M2T problem on $P$.

\begin{lemma}
\label{lem:optIsK}
Let $OPT$ be an optimal solution for the M2T problem on $P$. Then, $\lvert OPT\rvert\geq k$.
\end{lemma}
\begin{proof}
By Lemma~\ref{lem:standardOPT}, we assume that $OPT$ is in standard form; that is, $OPT\subseteq C(P)$ and every vertical 2-transmitter in $OPT$ is vertically maximal. Consider the partition $T=\{P_1, P_2,\dots, P_k\}$ of $P$ induced by the recursive steps of the algorithm, and let $s$ be a horizontal line segment in $P$. We say that $s$ \emph{originates} from $P_j$, for some $1\leq j\leq k$, if $\leftp{s}$ lies inside $P_j$. Suppose for a contradiction that $\lvert OPT\rvert<k$. Then, there must be a subpolygon $P_i\in T$ such that neither a vertical 2-transmitter of $OPT$ lies in $P_i$ nor a horizontal 2-transmitter of $OPT$ originates from $P_i$. We then must have one of the followings (w.l.o.g., we assume that Algorithm~\ref{alg:approximate2transmitters} located the pair $\{s_v, s_h\}$ in $P_i$):
\begin{itemize}
\item There exists at least one horizontal 2-transmitter in $OPT$ that intersects $\lefte{P_i}$ (and, therefore its left endpoint lies to the left of $\lefte{P_i}$). Let $s^*_h$ be the rightward maximal horizontal 2-transmitter among all such 2-transmitters. Clearly, $s^*_h$ does not see $P_i$ entirely because then \texttt{hVFinder($P$)} would have selected the portion of $s^*_h$ that lies in $P_i$ along with the vertical line segment $s'_v$ and so $P_i$ would have been extended further to the right. Now, let $P'_i := P_i\setminus \vis{s^*_h}$. Since $s^*_h$ is rightward maximal and there is no horizontal 2-transmitter of $OPT$ that is originated from $P_i$, we conclude that no horizontal 2-transmitter in $P$ sees a point in $P'_i$. Therefore, there must a vertical 2-transmitter $s^*_v$ that guards $P'_i$ and that $s^*_v$ lies to the left of $\lefte{P_i}$ or to the right of $\righte{P_i}$ (recall that there is no vertical 2-transmitter of $OPT$ inside $P_i$). \begin{inparaenum}[(i)] \item If $s^*_v$ lies to the right of $\righte{P_i}$, then our algorithm would have added $s^*_v$ and the portion of $s^*_h$ that lies in $P_i$ into $S$ and so $P_i$ would have been extended further to the right --- a contradiction. \item If $s^*_v$ lies to the left of $\lefte{P_i}$, then we observe that $s^*_v$ and $s_h$ (i.e., the horizontal 2-transmitter located in $P_i$ by our algorithm) would together guard $P_i$ entirely. This means that $\vis{s_v}\subseteq (\vis{s^*_v}\cup \vis{s_h})$ and so $s_v$ should have been removed from $C(P)$ in the first step of the algorithm --- a contradiction. \end{inparaenum}
\item There is no horizontal 2-transmitter of $OPT$ intersecting $\lefte{P_i}$. This means that no point inside $P_i$ is seen by a horizontal 2-transmitter in $P_i$. Moreover, since no vertical 2-transmitter of $OPT$ lies in $P_i$, we conclude that $P_i$ is guarded by a set $M\subseteq OPT$ of \emph{only-vertical} 2-transmitters that lie to the left of $\lefte{P_i}$ or to the right of $\righte{P_i}$. That is, $P_i \subseteq \bigcup_{s_j\in M} \vis{s_j}$. But, this means that $\vis{s_v}\subseteq \bigcup_{s_j\in M} \vis{s_j}$, which is a contradiction because then $s_v$ should have been removed from $C(P)$ in the first step of the algorithm.
\end{itemize}
By the two cases above, we conclude that $\lvert OPT\rvert \geq k$. This completes the proof of the lemma.
\end{proof}

Each call to methods \texttt{vHFinder($P$)} and \texttt{hVFinder($P$)} is completed in polynomial time. Moreover, the while-loop of Algorithm~\ref{alg:approximate2transmitters} terminates after at most $m$ iterations (recall that $m$ is the number of the vertical edges of $P$) because at least one new vertical edge of $P$ is guarded at each iteration. Therefore, Algorithm~\ref{alg:approximate2transmitters} runs in polynomial time. Therefore, by Lemma~\ref{lem:optIsK} and the fact that $\lvert S\rvert\leq 2k$, we have the main result of this paper:

\begin{theorem}
\label{thm:mainApproxResult}
There exists a polynomial-time 2-approximation algorithm for the M2T problem on monotone orthogonal polygons.
\end{theorem}


\section{Conclusion}
\label{sec:conclusion}
In this paper, we gave a polynomial-time 2-approximation algorithm for the M2T problem on monotone orthogonal polygons. The complexity of the problem remains open on simple orthogonal polygons. Similar to Katz and Morgenstern~\cite{katz2011}, it might be interesting to first consider the problem with only-vertical 2-transmitters.

\bibliographystyle{plain}
\bibliography{ref}

\end{document}